\documentclass[11pt]{article} \usepackage{amsmath} 
\usepackage{amsthm} 
\usepackage{amssymb} 
\usepackage{dsfont} 
\usepackage{mathrsfs}
\usepackage[pdftex,left=1in,top=1in,bottom=1in,right=1in]{geometry}

\title{Non-Asymptotic Capacity Upper Bounds for the Discrete-Time Poisson Channel with Positive Dark Current} \author{Mahdi Cheraghchi\thanks{Department of EECS, University of Michigan, Ann Arbor. Email: \texttt{mahdich@umich.edu}. Research supported in part by NSF grant CCF-2006455.}\and Jo\~ao Ribeiro\thanks{Department of Computing, Imperial College London. Email: \texttt{j.lourenco-ribeiro17@imperial.ac.uk}.}}
\date{}

 \newcommand{\R}{\mathbb{R}}
\newcommand{\E}{\mathds{E}} 
 
\newcommand{\eps}{\varepsilon}

\newcommand{\N}{\mathds{N}}

\newtheorem{thm}{Theorem} 
 \newtheorem{lem}[thm]{Lemma}
   \theoremstyle{definition}
 
\newtheorem{remark}[thm]{Remark}

\newcommand{\Yq}{Y^{(q)}}
\newcommand{\Yqd}{Y^{(q)}_\delta}
\newcommand{\Yql}{Y^{(q)}_\lambda}

\newcommand{\KL}{D_{\mathsf{KL}}}

\newcommand{\Poi}{\mathsf{Poi}}

\usepackage{color}

\usepackage{graphicx}

\usepackage{hyperref}

\begin{document}

\maketitle

\begin{abstract}
    We derive improved and easily computable upper bounds on the capacity of the discrete-time Poisson channel under an average-power constraint and an arbitrary constant dark current term.
    This is accomplished by combining a general convex duality framework with a modified version of the digamma distribution considered in previous work of the authors (Cheraghchi, J.\ ACM 2019; Cheraghchi, Ribeiro, IEEE Trans.\ Inf.\ Theory 2019).
    For most choices of parameters, our upper bounds improve upon previous results even when an additional peak-power constraint is imposed on the input.
\end{abstract}

\section{Introduction}\label{sec:intro}

The Discrete-Time Poisson (DTP) channel with \emph{dark current} $\lambda\geq 0$ is a memoryless channel which on input $x\in\R^+_0$ outputs $Y_x$ following a Poisson distribution with mean $\lambda+x$, which we denote by $\Poi_{\lambda+x}$.
In other words, it holds that
\begin{equation*}
    Y_x(y)=\Poi_{\lambda+x}(y)=e^{-(\lambda+x)}\frac{(\lambda+x)^y}{y!},\quad y=0,1,2,\dots,
\end{equation*}
where $Y_x(y)$ denotes the probability that the output of the DTP channel on input $x$ equals $y$.
This channel is a discrete model of optical communication first studied explicitly by Shamai~\cite{Sha90}, where the input $x$ describes the intensity of a photon-emitting source at the sender's side.
The receiver observes a photon count that follows a Poisson distribution, possibly corrupted by some background interference modelled by an additive dark current parameter $\lambda$.

Without any constraints, the capacity of the DTP channel is infinite.
However, in practice it is reasonable to impose some constraints on the input distributions for the DTP channel.
Two well-studied and practically motivated constraints are an average-power constraint $\mu$, where one restricts input distributions $X$ to those satisfying $\E[X]\leq \mu$, with $\mu=\infty$ meaning that no average-power constraint is imposed, and a peak-power constraint $A$, where one enforces that $\Pr[X\leq A]=1$, with $A=\infty$ meaning that no peak-power constraint is imposed.
Given such constraints on the input distribution $X$, we are interested in the capacity of the DTP channel with dark current $\lambda$ under an average- and/or peak-power constraint
\begin{equation*}
    C(\lambda,\mu,A)=\sup_{X:\E[X]\leq\mu,0\leq X\leq A}I(X;Y_X),
\end{equation*}
where $Y_X$ denotes the output distribution of the DTP channel with dark current $\lambda$ and input distribution $X$.
Throughout this work, we measure capacity in nats/channel use.
For simplicity, when $A=\infty$ (i.e., no peak-power constraint is imposed), we denote the corresponding capacity of the DTP channel by $C(\lambda,\mu)$, and furthermore when $\lambda=0$ (meaning that there is no dark current) we denote the corresponding capacity of the DTP channel by $C(\mu)$.
For every $\lambda$, $\mu$, and $A$, we have the chain of inequalities
\begin{equation*}
    C(\mu)\geq C(\lambda,\mu)\geq C(\lambda,\mu,A).
\end{equation*}

Currently, the exact value of $C(\lambda,\mu,A)$ is not known for any non-trivial choice of parameters, although we have some upper and lower bounds on this quantity along with some asymptotic results.
Moreover, when $A$ is finite, we also know algorithms for numerically approximating the capacity~\cite{CHC13,SSEL15}.
We discuss these in detail in Section~\ref{sec:prev}.
Notably, whenever the average- and peak-power constraints $\mu$ and $A$ are neither very small nor very large, and whenever the dark current $\lambda$ is not very large compared to $\mu$ and $A$, the best known analytical upper bound on $C(\lambda,\mu,A)$ is actually an upper bound on $C(\mu)$~\cite{CR19a}.

\paragraph{Our contributions.}
In this work, we derive significantly improved non-asymptotic upper bounds on $C(\lambda,\mu)$, and thus also on $C(\lambda,\mu,A)$ for most reasonable choices of parameters when the dark current $\lambda>0$ is constant.
Our upper bounds are easy to compute and are in turn sharply upper bounded by closed-form, elementary expressions.

\subsection{Previous work}\label{sec:prev}

In this section, we discuss previous results on the capacity of the constrained DTP channel, with special focus on known (asymptotic and non-asymptotic) capacity upper bounds.
Most previous work has focused on asymptotic settings where $\mu\to 0$ or $\mu\to\infty$, although there exist some capacity bounds applicable to non-asymptotic settings.

\paragraph{Capacity upper bounds.}
Brady and Verdú~\cite{BV90,Bra90} were the first to study the asymptotic capacity of the DTP channel when $\mu\to\infty$.
They derived bounds on the capacity of the DTP channel with an average-power constraint only, $C(\lambda,\mu)$, when $\mu\to\infty$ and the ratio $\mu/\lambda$ stays constant (meaning, in particular, that $\lambda\to\infty$ as well).
Namely, for every $\eps>0$ they derived the asymptotic upper bound (see~\cite[Section 4, Proof of Theorem 4]{Bra90})
\begin{equation}\label{eq:BVbound}
    C(\lambda,\mu)\leq \ln(1+\mu+\lambda)+(\mu+\lambda)\ln\left(1+\frac{1}{\mu+\lambda}\right)-\frac{1}{2}\ln(2\pi(\mu+\lambda))+\ln(3/2)+\eps
\end{equation}
valid for all $\mu>C_\eps$, where $C_\eps$ is a large constant depending on $\eps>0$.
A characterization of the asymptotic behavior of $C(\lambda,\mu)$ when $\mu\to\infty$ and $\lambda$ is constant was later obtained by Martinez~\cite{Mar07} and Lapidoth and Moser~\cite{LM09}, who showed that
\begin{equation*}
    \lim_{\mu\to\infty}\frac{C(\lambda,\mu)}{\ln\mu}=\frac{1}{2},
\end{equation*}
when $\lambda\geq 0$ is an arbitrary constant.
Overall, the best upper bound on $C(\mu)$ for any $\mu$ outside the asymptotic regime $\mu\to 0$ was obtained by the authors in~\cite{CR19a}, improving on a previous upper bound of Martinez~\cite{Mar07}, and is given by
\begin{equation}\label{eq:best0dark}
    C(\mu)\leq \mu  \ln \left(\frac{1+\left(1+e^{1+\gamma }\right) \mu+2\mu ^2}{e^{1+\gamma } \mu + 2 \mu ^2}\right)
	+\ln \left(1+\frac{1}{\sqrt{2e}}\left(\sqrt{\frac{1+(1+e^{1+\gamma})\mu+2\mu^2}{1+\mu}}-1\right)\right),
\end{equation}
where $\gamma\approx 0.5772$ is the Euler-Mascheroni constant.
Notably, the bound in~\eqref{eq:best0dark} has the same first-order asymptotic behavior as $C(\mu)$ both when $\mu\to\infty$, and, as we discuss below, when $\mu\to 0$.

The setting where $\mu\to 0$ was first considered by Lapidoth, Shapiro, Venkatesan, and Wang~\cite{LSVW11}, who determined the first order behavior of $C(\lambda,\mu,A)$ in this setting, both when $\lambda=0$ and $\lambda>0$ is a fixed constant, and both with and without a peak-power constraint $A$.
When there is no dark current $\lambda$ present, it holds that
\begin{equation*}
    \lim_{\mu\to 0}\frac{C(\mu,A)}{\mu\ln(1/\mu)}=1.
\end{equation*}
On the other hand, when $\lambda>0$ is constant, we have
\begin{equation}\label{eq:muto0darkcurr}
    \frac{1}{2}\leq \liminf_{\mu\to 0}\frac{C(\lambda,\mu)}{\mu\ln\ln(1/\mu)}\leq\limsup_{\mu\to 0}\frac{C(\lambda,\mu)}{\mu\ln\ln(1/\mu)}\leq 2.
\end{equation}
In order to prove~\eqref{eq:muto0darkcurr}, the authors~\cite[Expression (114)]{LSVW11} derive an explicit non-asymptotic upper bound on $C(\lambda,\mu)$ given by
\begin{equation}\label{eq:lapidothbound}
    C(\lambda,\mu)\leq F_1(\lambda,\mu)+F_2(\lambda,\mu)+F_3(\lambda,\mu)
\end{equation}
with $F_1$, $F_2$, and $F_3$ given by
\begin{align*}
    F_1(\lambda,\mu)&=\left(\eta\ln\eta+\frac{1}{12\eta}+\frac{1}{2}\ln(2\pi\eta)+\lambda-\eta\ln\lambda-\ln(1-p)\right)e^{\eta+\eta\ln\lambda-\eta\ln\eta+\frac{\mu}{\eta-\sqrt{\eta}-\lambda}},\\
    F_2(\lambda,\mu)&=\max\left(0,\left(1+\ln(1/p)+\ln\lambda\right)\left(\mu+\frac{\lambda\mu}{\eta-\sqrt{\eta}-\lambda}+\lambda e^{\eta-1-\lambda+(\eta-1)\ln\lambda-(\eta-1)\ln(\eta-1)}\right)\right),\\
    F_3(\lambda,\mu)&=\mu\left(1+\frac{\lambda}{\eta-\lambda}\right)\max\left(0,\ln(1/\lambda)\right)+\mu\frac{\eta\ln(\eta/\lambda)}{\eta-\lambda},
\end{align*}
where $\eta$ is a free parameter that must be larger than some constant $C_\lambda>0$ depending on $\lambda$ and $p\in(0,1)$ is a free parameter.
By inspection of~\cite[Section IV-B]{LSVW11}, it must at the very least be the case that $\eta-\sqrt{\eta}>\lambda$ for the bound to hold.
Therefore, the upper bound in~\eqref{eq:lapidothbound} is always significantly larger than
\begin{equation}\label{eq:lapidothunderestimate}
    \mu(1+\max(0,1+\ln\lambda)+\max(0,\ln(1/\lambda))).
\end{equation}
We will use this conservative underestimate of~\eqref{eq:lapidothbound} when comparing the different bounds in Section~\ref{sec:compbounds}.
Later, Wang and Wornell~\cite{WW14} determined the second-order asymptotics of the capacity of the DTP channel under an average-power constraint $\mu$ with dark current $\lambda=c\mu$ for an arbitrary constant $c>0$.
They showed that
\begin{equation*}
    C(c\mu,\mu)=\mu\ln(1/\mu)+\mu\ln\ln(1/\mu)+O_c(\mu)
\end{equation*}
when $\mu\to 0$, where the term $O_c(\mu)$ depends on the constant $c$.
Moreover, they gave an upper bound~\cite[Expression (180)]{WW14} on $C(\lambda,\mu)$ matching this asymptotic behavior which holds whenever $\mu$ and $\lambda$ are small enough,
\begin{multline}\label{eq:WWbound}
    C(\lambda,\mu)\leq \mu\ln\ln(1/\mu)+\mu-\log(1-\mu-\lambda)-\lambda+\frac{\lambda^2}{2}\ln\ln(1/\mu)\\-(\mu+\lambda)\log\left(1-\frac{1}{\ln(1/\mu)}\right)+\mu e^{-\lambda}\sup_{x\geq 0}\phi_{\mu,\lambda}(x),
\end{multline}
where $\phi_{\mu,\lambda}(x):=\frac{1-e^{-x}}{x}\ln\left(\frac{x+\lambda}{(\mu+\lambda)\ln(1/\mu)}\right)$.

Finally, Aminian et al.~\cite[Example 2]{AAGNKM15} also derived a non-asymptotic upper bound on $C(\lambda,\mu,A)$ given by
\begin{equation}\label{eq:cov}
    C(\lambda,\mu,A)\leq \sup_{X:\E[X]\leq \mu,X\leq A}\textrm{Cov}(X+\lambda,\ln(X+\lambda))=\begin{cases}
\frac{\mu}{A}(A-\mu)\ln(A/\lambda+1),&\text{ if $\mu<A/2$,}\\
\frac{A}{4}\ln(A/\lambda+1),&\text{ otherwise,}
\end{cases}
\end{equation}
where $\textrm{Cov}$ denotes the covariance.
Note that when $A\to\infty$, the upper bound in~\eqref{eq:cov} becomes arbitrarily large.
Therefore, it does not imply any non-trivial upper bound on $C(\lambda,\mu)$.
However, this bound may be used to recover some known asymptotic results on $C(\lambda,\mu,A)$ when $\mu\to 0$ from~\cite{LSVW11}.

\paragraph{Other results on the DTP channel.}
Besides the capacity upper bounds discussed above, several other aspects of the capacity of the DTP channel have been studied.
Still with respect to capacity bounds, several works have derived both asymptotic and non-asymptotic lower bounds on the capacity of the DTP channel under different combinations of average- and peak-power constraints and dark current~\cite{Gor62,GB65,His71,JM71,Mar07,LM09,CHC10,LSVW11,WW14,YZWD14}.
Moreover, when both average- and peak-power constraints are imposed on the input, some algorithms have been proposed to numerically approximate $C(\lambda,\mu,A)$ and the corresponding capacity-achieving input distribution~\cite{CHC13,SSEL15}.

On another note, the properties of the capacity-achieving distribution for the DTP channel have also been studied.
Shamai~\cite{Sha90} was the first to study this problem, and showed that the support of the capacity-achieving distribution is finite when both average- and peak-power constraints are imposed. Moreover, he also gave conditions on $\mu$ and $A$ that ensure that an input distribution with two or three mass points is capacity-achieving.
Later, Cao, Hranilovic, and Chen~\cite{CHC14a,CHC14b} extended the results of Shamai.
In particular, they showed that the capacity-achieving distribution under both average- and peak-power constraints must have some probability mass at $x=0$.
Moreover, if there is only an active peak-power constraint, they show there must also be some probability mass at $x=A$, and that this may not be the case otherwise.
Additionally, they showed that the support of the capacity-achieving distribution under an average-power constraint only must be unbounded.
This result was then strengthened in~\cite{CR19a}, where it is shown that the capacity-achieving distribution under an average-power constraint only has countably infinite support, with a finite number of mass points in every bounded interval.

\subsection{Notation}
We denote random variables by uppercase letters such as $X$, $Y$, and $Z$.
For a discrete random variable $X$, we denote by $X(x)$ the probability that $X$ equals $x$, and the expected value of $X$ is denoted by $\E[X]$.
Moreover, we denote the Shannon entropy of a discrete random variable $X$ by $H(X)$, and the Kullback-Leibler divergence between two discrete random variables $X$ and $Y$ by $\KL(X\|Y)$.
The natural logarithm is denoted by $\ln$, and we measure capacity in nats/channel use.

\subsection{Organization}
We describe our new upper bound on $C(\lambda,\mu)$ along with its proof in Section~\ref{sec:main}.
Then, we compare our upper bound with previously known bounds on $C(\lambda,\mu)$ and $C(\lambda,\mu,A)$ in Section~\ref{sec:compbounds}.

\section{The main result}\label{sec:main}

In this paper, we prove the following theorem, which yields a significantly improved non-asymptotic upper bound on the capacity of the DTP channel with constant dark current $\lambda$ and an average-power constraint $\mu$ in non-asymptotic regimes of $\mu$.
\begin{thm}\label{thm:UB}
For every $\mu,\lambda\geq 0$ we have
\begin{equation}\label{eq:UBdarkcurrent}
    C(\lambda,\mu)\leq \ln\left(\delta_\lambda+\frac{1}{\sqrt{2e}}\left(\frac{1}{\sqrt{1-q_{\lambda,\mu}}}-1\right)\right)-(\mu+\lambda)\ln q_{\lambda,\mu},
\end{equation}
with
\begin{equation*}
    \delta_\lambda=\exp(-\lambda e^\lambda E_1(\lambda)),
\end{equation*}
where $E_1(z)=\int_1^\infty\frac{e^{-zt}}{t}dt$ is the exponential integral function (with the convention that $0E_1(0)=0$), and $q_{\lambda,\mu}$ defined as
\begin{equation*}
    q_{\lambda,\mu}=1-\frac{1}{1+e^{1+\gamma}(\mu+\lambda)+\frac{2-e^{1+\gamma}}{1+\mu+\lambda}(\mu+\lambda)^2},
\end{equation*}
where $\gamma\approx0.5772$ is the Euler-Mascheroni constant.
\end{thm}

\begin{remark}
Although the upper bound from Theorem~\ref{thm:UB} does not have a closed-form expression since it features the exponential integral function (which is nevertheless easy to compute numerically), we can derive a good closed-form and elementary upper bound on $C(\mu,\lambda)$ by noting that~\cite[Section 5.1.20]{AS65} and~\cite[Theorem 2]{Alz97} give the lower bound
\begin{equation*}
    e^\lambda E_1(\lambda)\geq \max\left(\frac{1}{2}\ln(1+2/\lambda),-e^\lambda\ln\left(1-e^{-\lambda e^\gamma}\right)\right)
\end{equation*}
for all $\lambda>0$, and thus
\begin{equation}\label{eq:elemUBdelta}
    \delta_\lambda\leq \min\left((1+2/\lambda)^{-\lambda/2},(1-e^{-\lambda e^\gamma})^{\lambda e^\lambda}\right),
\end{equation}
with the right hand side expression in the minimum above being better for small $\lambda$ (e.g., $\lambda<1/2$).
When $\lambda$ is small, this elementary upper bound sharply approaches the upper bound from Theorem~\ref{thm:UB}, and overall it improves on previously known bounds whenever $\mu$ is not small compared to $\lambda$.
\end{remark}

As with most previous capacity upper bounds for the DTP channel, we derive Theorem~\ref{thm:UB} with the help of a general convex duality framework, which we state below in a specialized form for the DTP channel.
This framework was originally derived in~\cite{Che19} and has also been used to derive the state-of-the-art upper bound on $C(\mu)$~\cite{CR19a}.
As discussed in~\cite{CR19a}, it is equivalent to other existing frameworks (e.g., see~\cite{LM03,Mar07}).
\begin{lem}[\cite{Che19,CR19a}]\label{lem:duality}
Suppose that there exist constants $a\in\R^+_0$, $b\in \R$, and a distribution $Y$ supported in $\N$ such that
\begin{equation}\label{eq:KLY}
    \KL(\Poi_z\|Y)\leq az+b
\end{equation}
for all $z\geq \lambda$, where $\KL(\Poi_z\|Y)$ denotes the Kullback-Leibler divergence between the Poisson distribution with mean $z$, denoted by $\Poi_z$, and $Y$.
Then, we have
\begin{equation*}
    C(\lambda,\mu)\leq a(\mu+\lambda)+b
\end{equation*}
for all $\lambda,\mu\geq 0$.
\end{lem}

An important quantity related to Lemma~\ref{lem:duality} is the \emph{KL-gap} of the distribution $Y$ with respect to the line $az+b$, which quantifies the sharpness with which the constraint~\eqref{eq:KLY} is satisfied, and is defined as
\begin{equation*}
    \Delta(z)=az+b-\KL(\Poi_z\|Y).
\end{equation*}
From previous applications of Lemma~\ref{lem:duality}~\cite{Che19,CR19a,CR19b}, it is apparent that designing candidate distributions $Y$ so that the associated KL-gap is as small as possible leads to sharper capacity upper bounds.
We follow this approach in this work as well.

\subsection{The digamma distribution}
To give some context, we start by presenting the family of digamma distributions originally introduced in~\cite{Che19} in order to derive capacity upper bounds for channels with synchronization errors, and which was also used to derive the state-of-the-art upper bounds on $C(\mu)$ in~\cite{CR19a}.
This is a family of distributions $\Yq$ parameterized by $q\in(0,1)$, and each distribution $\Yq$ satisfies
\begin{equation}\label{eq:digammadist}
    \Yq(y)=y_0q^y \frac{e^{g(y)-y}}{y!},\quad y=0,1,2,\dots,
\end{equation}
with $y_0$ the normalizing factor and $g(y)$ defined as
\begin{equation*}
    g(y)=\begin{cases}
    y\psi(y),&\text{ if $y>0$,}\\
    0,&\text{ if $y=0$,}
    \end{cases}
\end{equation*}
where $\psi$ denotes the digamma function, which for positive integer argument $y$ satisfies
\begin{equation*}
    \psi(y)=-\gamma+\sum_{i=1}^{y-1}1/i,
\end{equation*}
where we recall that $\gamma\approx 0.5772$ is the Euler-Mascheroni constant.
Using well-known properties of some special functions, it is possible to compute $\KL(\Poi_z\|\Yq)$ exactly, as stated in the following lemma.
\begin{lem}[\cite{Che19,CR19a}]\label{lem:origKL}
For every $z\geq 0$ and $q\in(0,1)$ it holds that
\begin{equation*}
    \KL(\Poi_z\|\Yq)=-\ln y_0-z\ln q - zE_1(z),
\end{equation*}
where we recall that $E_1(z)=\int_1^\infty \frac{e^{-zt}}{t}dt$ is the exponential integral function, with the convention that $0\cdot E_1(0)=0$.
\end{lem}
From Lemma~\ref{lem:origKL}, we conclude that each distribution $\Yq$ satisfies
\begin{equation}\label{eq:KLdigamma}
    \KL(\Poi_z\|\Yq)\leq -\ln y_0-z\ln q
\end{equation}
for all $z\geq 0$, with associated KL-gap $\Delta(z)$ satisfying
\begin{equation}\label{eq:KLgapdigamma}
    \Delta(z)=-\ln y_0-z\ln q-\KL(\Poi_z\|\Yq)=zE_1(z)\geq 0.
\end{equation}
Note that $\Delta(z)$ is small around $z=0$ (with $\Delta(0)=0$) and $\Delta(z)\to 0$ exponentially fast as $z\to\infty$.
Combining~\eqref{eq:KLdigamma} with Lemma~\ref{lem:duality}, we conclude that
\begin{equation}\label{eq:infUBdigamma}
    C(\mu)\leq \inf_{q\in(0,1)}[-\ln y_0-\mu\ln q].
\end{equation}
Then, choosing $q$ appropriately as a function of $\mu$ and using some known upper bounds on $-\ln y_0$ as a function of $q$ yields a good easy-to-compute closed-form upper bound on $C(\mu)$.

Given the above, it is natural to wonder whether one can apply the digamma distribution in a straightforward way to obtain better bounds on $C(\lambda,\mu)$ for $\lambda>0$.
However, this cannot be done since, with Lemma~\ref{lem:duality} in view, we have $a(\mu+\lambda)+b>a\mu+b$ for any $a\in\R^+_0$ and $b\in\R$.
Therefore, any upper bound obtained for $C(\lambda,\mu)$ via the digamma distributions above will be strictly larger than~\eqref{eq:infUBdigamma}, and thus trivial. 
In the next section, we show how to modify the family of digamma distributions as a function of the dark current $\lambda$ in order to obtain significantly improved easy-to-compute closed-form upper bounds on $C(\lambda,\mu)$.

\subsection{The modified digamma distribution}

In this section, we design and study a modified version of the family of digamma distributions defined in~\eqref{eq:digammadist}.
Our modification consists in changing the value of the digamma distribution $\Yq$ at $y=0$ and renormalizing the distribution.
More precisely,
for $\delta\in(0,1]$ we consider the \emph{modified digamma distribution} $\Yq_\delta$ defined as
\begin{equation*}
    \Yq_\delta(y)=\begin{cases}
    \alpha\delta,&\text{ if $y=0$,}\\
    \alpha\Yq(y)/y_0,&\text{ if $y>0$,}
    \end{cases}
\end{equation*}
where $\alpha$ is the new normalizing factor satisfying
\begin{equation*}
    1/\alpha=1/y_0+\delta-1,
\end{equation*}
where we have used the fact that $\Yq(0)/y_0=1$.
An analogous approach was used by the authors in~\cite{CR19b} to derive improved capacity upper bounds on channels with synchronization errors. Moreover, we note that a similar approach was employed by Martinez~\cite{Mar07} in the special case where $\lambda=0$ to improve the upper bound given by his candidate distribution, which originally had KL-gap bounded well away from $0$ everywhere.
However, no rigorous proof is given in~\cite{Mar07} to show that this approach indeed works in that special case, with only numerical evidence being presented.

We begin by computing $\KL(\Poi_z\|\Yqd)$ for general $q\in(0,1)$ and $\delta\in(0,1]$, which has a simple expression in terms of the original KL-gap $\Delta$ of the digamma distribution $\Yq$ defined in~\eqref{eq:KLgapdigamma}.
We have
\begin{align}
    \KL(\Poi_z\|\Yqd)&=-H(\Poi_z)-\sum_{y=0}^\infty \Poi_z(y)\ln \Yqd(y)\nonumber\\
    &=-\ln\alpha-\Poi_z(0)\ln\delta-H(\Poi_z)+\sum_{y=1}^\infty \Poi_z(y)(\ln(y!)+y-g(y))\nonumber\\
    &=-\ln\alpha-\Poi_z(0)\ln\delta-H(\Poi_z)+\E_{y\sim\Poi_z}[\ln(y!)+y-g(y)-y\ln q]\nonumber\\
    &=-\ln\alpha-z\ln q-e^{-z}\ln\delta-H(\Poi_z)+\E_{y\sim\Poi_z}[\ln(y!)+y-g(y)]\nonumber\\
    &=-\ln\alpha-z\ln q-e^{-z}\ln\delta-zE_1(z).\label{eq:KLdelta}
\end{align}
The third equality holds because the term inside the sum is $0$ at $y=0$.
The fourth equality is true since $\E[\Poi_z]=z$ and $\Poi_z(0)=e^{-z}$.
The fifth equality follows from the fact that
\begin{equation*}
    zE_1(z)=\Delta(z)=H(\Poi_z)-\E_{y\sim\Poi_z}[\ln(y!)+y-g(y)].
\end{equation*}

Given $\lambda\geq 0$, consider now the choice
\begin{equation}\label{eq:deltalambda}
    \delta_\lambda=\exp(-\lambda e^\lambda E_1(\lambda)).
\end{equation}
Then, we have
\begin{equation*}
    -e^{-z}\ln\delta_\lambda= \lambda e^{\lambda-z}E_1(\lambda).
\end{equation*}
Consequently, by defining $\Yq_\lambda=\Yq_{\delta_\lambda}$ and using~\eqref{eq:KLdelta} we have
\begin{equation}\label{eq:KLdeltalambda}
    \KL(\Poi_z\|\Yq_\lambda)=-\ln\alpha-z\ln q+\lambda e^{\lambda-z}E_1(\lambda)-zE_1(z).
\end{equation}
We now claim that the following result holds.
\begin{thm}\label{thm:moddigamma}
For every $z\geq \lambda$ and $q\in(0,1)$ we have
\begin{equation*}
    \KL(\Poi_z\|\Yq_\lambda)\leq -\ln\alpha-z\ln q,
\end{equation*}
with KL-gap $\Delta_\lambda$ satisfying
\begin{equation*}
    \Delta_\lambda(z)=-\ln\alpha-z\ln q-\KL(\Poi_z\|\Yq_\lambda)=zE_1(z)-\lambda e^{\lambda-z}E_1(\lambda)< \Delta(z).
\end{equation*}
\end{thm}
Note that we have $\Delta_\lambda(\lambda)=0$ and $\Delta_\lambda(z)\to 0$ exponentially fast when $z\to\infty$ (in general, $\Delta_\lambda$ is always smaller than $\Delta$).
Theorem~\ref{thm:moddigamma} and the observations above justify our choice of $\delta_\lambda$ in~\eqref{eq:deltalambda}; With this choice, we obtain a new family of modified digamma distributions $\Yql$ with KL-gap $\Delta_\lambda$ that is always smaller than the original KL-gap $\Delta$ of the digamma distributions.
Moreover, the KL-gap $\Delta_\lambda$ equals $0$ at $z=\lambda$ and is significantly smaller than $\Delta$ around $z=\lambda$.
Given the above, intuitively we expect to obtain a sharper upper bound on $C(\lambda,\mu)$ using the family of modified digamma distributions.

Finally, Theorem~\ref{thm:moddigamma} is an immediate consequence of~\eqref{eq:KLdeltalambda} and the following lemma.
\begin{lem}\label{lem:infklgap}
For every $z\geq \lambda$ we have
\begin{equation*}
    z E_1(z)- e^{\lambda-z}\lambda E_1(\lambda)\geq 0.
\end{equation*}
\end{lem}
\begin{proof}
Multiplying both sides of the inequality above by $e^{z}$, we conclude that the desired inequality holds provided we can show that
\begin{equation*}
    z e^z E_1(z)\geq \lambda e^\lambda E_1(\lambda)
\end{equation*}
for all $z\geq \lambda$.
Equivalently, we must show that the function $f(z)=z e^z E_1(z)$ is non-decreasing when $z>0$.
Note that we have
\begin{equation*}
    f'(z)=(1+z)e^z E_1(z)-1
\end{equation*}
for every $z>0$, and we proceed to show that $f'(z)\geq 0$ for all $z>0$.
This implies the desired result.
According to~\cite[Section 5.1.20]{AS65}, we can lower bound $e^z E_1(z)$ as
\begin{equation*}
    e^z E_1(z)>\frac{1}{2}\ln(1+2/z)
\end{equation*}
for all $z>0$.
Therefore, in order to show that $f'(z)\geq 0$ it is enough to argue that
\begin{equation*}
    \frac{1+z}{2}\cdot\ln(1+2/z)\geq 1
\end{equation*}
for all $z>0$.
This follows from the fact that $\ln(1+x)\geq \frac{2x}{2+x}$ for all $x\geq 0$, and thus
\begin{align*}
    \frac{1+z}{2}\cdot\ln(1+2/z)\geq \frac{1+z}{2}\cdot\frac{4/z}{2+2/z}=1. \tag*{\qedhere}
\end{align*}
\end{proof}

\subsection{Proof of Theorem~\ref{thm:UB}}

In this section, we prove our main result (Theorem~\ref{thm:UB}) with the help of Lemma~\ref{lem:duality} and Theorem~\ref{thm:moddigamma}.
First, by combining Lemma~\ref{lem:duality} and Theorem~\ref{thm:moddigamma} we conclude that
\begin{equation}\label{eq:UB1lambda}
    C(\lambda,\mu)\leq \inf_{q\in(0,1)}[-\ln \alpha-(\mu+\lambda)\ln q].
\end{equation}
To obtain Theorem~\ref{thm:UB} from~\eqref{eq:UB1lambda}, we upper bound the term $-\ln\alpha$ by an easy-to-compute expression in terms of $\lambda$ and $q$, and then choose $q$ appropriately as a function of $\lambda$ and $\mu$.

From~\cite{Che19}, we have the following upper bound on $1/y_0$, where $y_0$ is the normalizing factor of the digamma distribution $\Yq$.
\begin{lem}[\cite{Che19}]
For every $q\in (0,1)$ we have
\begin{equation*}
    1/y_0\leq 1+\frac{1}{\sqrt{2e}}\left(\frac{1}{\sqrt{1-q}}-1\right).
\end{equation*}
\end{lem}
Recalling that $1/\alpha=1/y_0-1+\delta_\lambda$, we conclude that
\begin{equation}\label{eq:UBalpha}
    -\ln\alpha\leq \ln\left(\delta_\lambda+\frac{1}{\sqrt{2e}}\left(\frac{1}{\sqrt{1-q}}-1\right)\right)
\end{equation}
for every $q\in(0,1)$.
It remains now to choose $q=q_{\lambda,\mu}$ appropriately.
As discussed in~\cite{CR19a}, when $\lambda=0$ the choice
\begin{equation*}
    q_\mu=1-\frac{1}{1+e^{1+\gamma}\mu+\frac{2-e^{1+\gamma}}{1+\mu}\mu^2}
\end{equation*}
is close to optimal for all $\mu$, and leads to a significantly improved (and closed-form) upper bound on $C(\mu)$.
For the case where $\lambda>0$, we consider the direct extension $q_{\lambda,\mu}$ defined as
\begin{equation}\label{eq:qmulambda}
    q_{\lambda,\mu}=1-\frac{1}{1+e^{1+\gamma}(\mu+\lambda)+\frac{2-e^{1+\gamma}}{1+\mu+\lambda}(\mu+\lambda)^2}.
\end{equation}
Combining~\eqref{eq:UB1lambda},~\eqref{eq:UBalpha}, and~\eqref{eq:qmulambda} leads to Theorem~\ref{thm:UB}.

\section{Comparison between bounds}\label{sec:compbounds}

We present a comparison between the upper bound~\eqref{eq:UBdarkcurrent} that we have derived via the modified digamma distribution and previously known upper bounds on $C(\lambda,\mu)$ and $C(\lambda,\mu,A)$ in Figures~\ref{fig:compboundbest} and~\ref{fig:compboundall}.
As can be observed, the upper bound~\eqref{eq:UBdarkcurrent} significantly improves on previous upper bounds whenever $\mu$ is not small compared to $\lambda$, and the elementary upper bound obtained by replacing $\delta_\lambda$ with the upper bound from~\eqref{eq:elemUBdelta} sharply approaches~\eqref{eq:UBdarkcurrent}.

\begin{figure}
	\centering
	\includegraphics[width=0.8\textwidth]{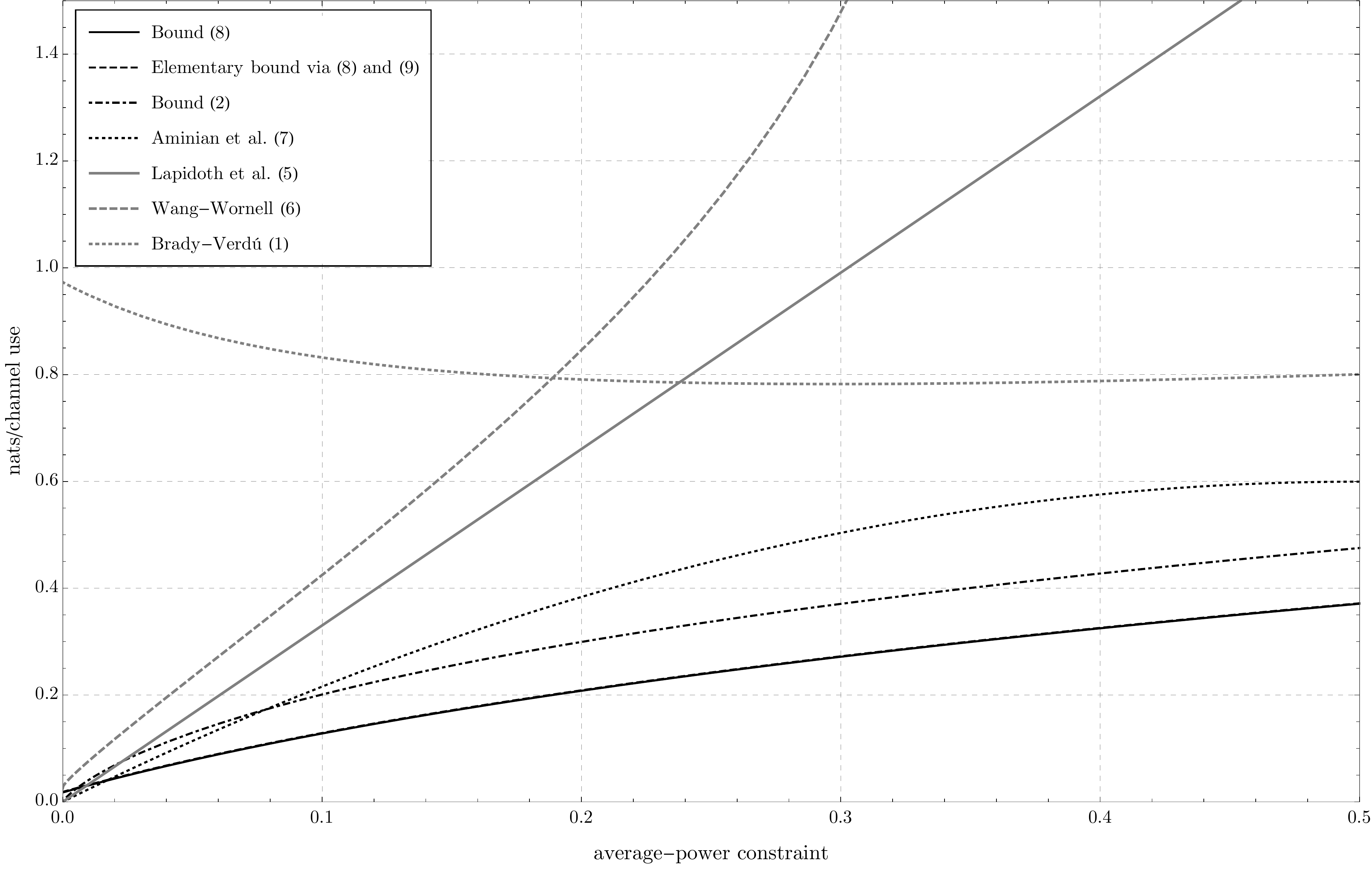}
	\caption{Comparison between the upper bound~\eqref{eq:UBdarkcurrent}, the elementary upper bound obtain by combining~\eqref{eq:UBdarkcurrent} and~\eqref{eq:elemUBdelta}, and previous upper bounds when $\lambda=1/10$. The upper bound~\eqref{eq:lapidothbound} is replaced by the underestimate~\eqref{eq:lapidothunderestimate}, the upper bound~\eqref{eq:cov} is computed assuming a peak-power constraint $A=1$, the upper bound~\eqref{eq:WWbound} is plotted without the additve $\mu e^{-\lambda}\sup_{x\geq 0}\phi_{\mu,\lambda}(x)$ term, which is always positive when it is well-defined, and the upper bound~\eqref{eq:BVbound} is plotted by ignoring the positive asymptotic term $\eps$.}
	\label{fig:compboundall}
\end{figure}

\begin{figure}
	\centering
	\includegraphics[width=0.8\textwidth]{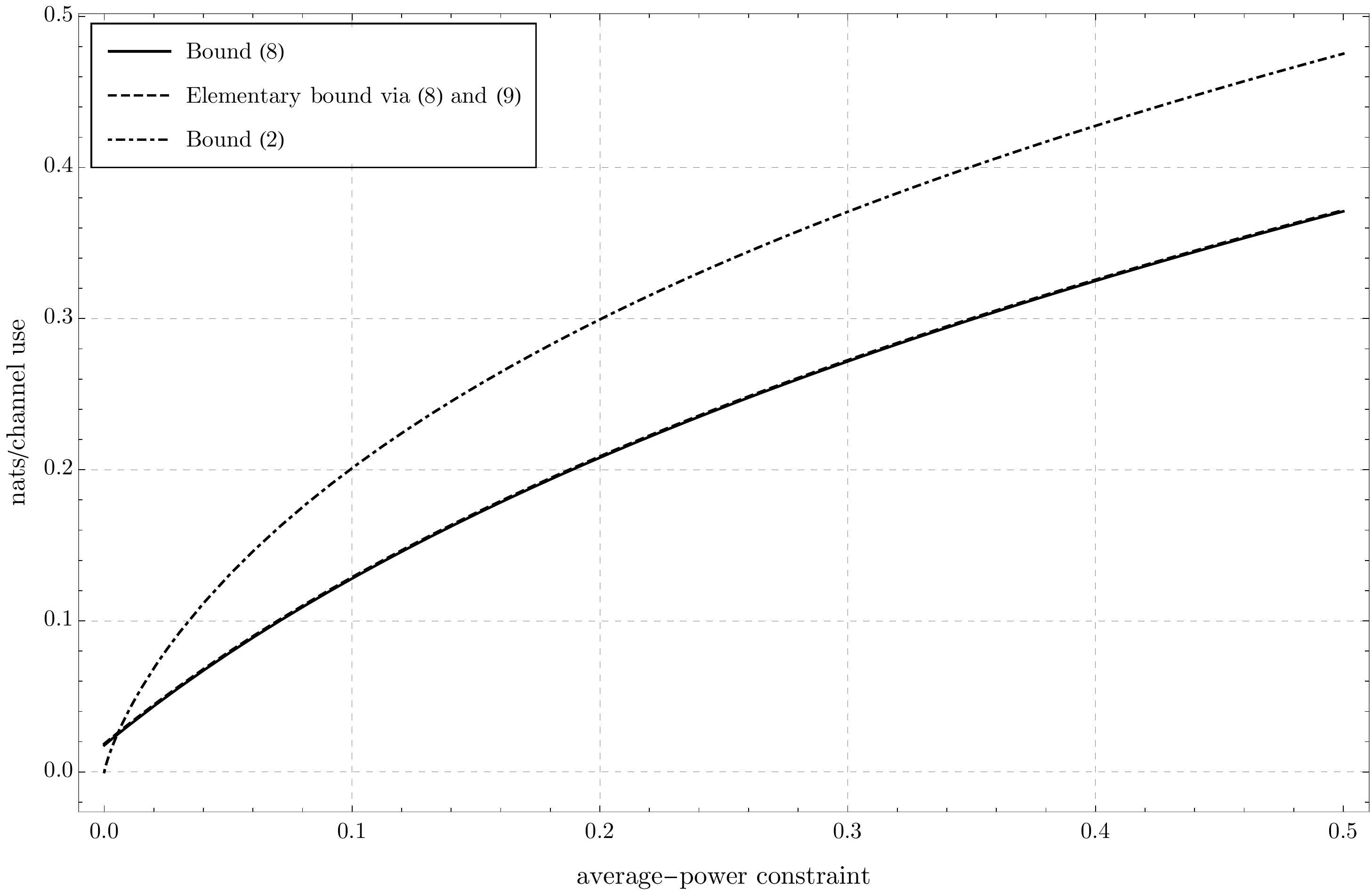}
	\caption{Comparison between the upper bound~\eqref{eq:UBdarkcurrent}, the elementary upper bound obtain by combining~\eqref{eq:UBdarkcurrent} and~\eqref{eq:elemUBdelta}, and the upper bound~\eqref{eq:best0dark} when $\lambda=1/10$.}
	\label{fig:compboundbest}
\end{figure}

\section*{Acknowledgments}

We thank Jun Chen for an insightful discussion regarding the results from~\cite{CHC14b}.

\bibliographystyle{IEEEtran}
\bibliography{refs}

\end{document}